\documentclass[11pt,letterpaper]{article}

\usepackage{mathpazo}

\usepackage{authblk}

\usepackage{amsmath, amsthm, amssymb,color,fullpage,url,booktabs}  
\usepackage[pdftex]{hyperref} 
\usepackage{graphicx}

\newcommand{\nc}{\newcommand}
\nc{\rnc}{\renewcommand}

\newcommand{\bra}[1]{\left\langle #1\right|}
\newcommand{\ket}[1]{\left|#1\right\rangle}

\DeclareMathOperator{\poly}{poly}
\DeclareMathOperator{\polylog}{polylog}
\DeclareMathOperator{\tr}{tr}

\DeclareMathOperator{\ORACLE}{\mathsf{ORACLE}}

\DeclareMathOperator{\GS}{\mathsf{GibbsSampler}}

\def\be#1\ee{\begin{equation}#1\end{equation}}
\def\bea#1\eea{\begin{eqnarray}#1\end{eqnarray}}
\def\beas#1\eeas{\begin{eqnarray*}#1\end{eqnarray*}}
\def\ba#1\ea{\begin{align}#1\end{align}}
\def\bas#1\eas{\begin{align*}#1\end{align*}}
\def\bpm#1\epm{\begin{pmatrix}#1\end{pmatrix}}

\newtheorem{thm}{Theorem}
\newtheorem*{thm*}{Theorem}

\newtheorem{cor}[thm]{Corollary}
\newtheorem{lem}[thm]{Lemma}

\newtheorem{dfn}[thm]{Definition}
\newtheorem{proto}{Protocol}

\makeatletter
\newtheorem*{rep@theorem}{\rep@title}
\newcommand{\newreptheorem}[2]{%
\newenvironment{rep#1}[1]{%
 \def\rep@title{#2 \ref{##1} (restatement)}%
 \begin{rep@theorem}}%
 {\end{rep@theorem}}}
\makeatother

\newreptheorem{thm}{Theorem}
\newreptheorem{lem}{Lemma}

\def\benum{\begin{enumerate}}
\def\eenum{\end{enumerate}}
\def\bit{\begin{itemize}}
\def\eit{\end{itemize}}

\nc{\todo}[1]{\textcolor{red}{todo: #1}}

\def\begsub#1#2\endsub{\begin{subequations}\label{eq:#1}\begin{align}#2\end{align}\end{subequations}}
\nc\qand{\qquad\text{and}\qquad}
\nc\mnb[1]{\medskip\noindent{\bf #1}}

\usepackage{boxedminipage}
\newtheorem{algorithm}[thm]{Algorithm}
\newenvironment{mybox}
{\center \noindent\begin{boxedminipage}{1.0\linewidth}}
{\end{boxedminipage}\noindent}

\usepackage{amsmath}
\usepackage{amssymb}

\usepackage[T1]{fontenc}
\usepackage[utf8]{inputenc}
\usepackage{authblk}

\begin{document}

\title{{\huge Quantum Speed-ups for Semidefinite Programming}}

\author[1]{Fernando G.S.L. Brand\~ao}
\author[2]{Krysta M.~Svore}

\affil[1]{Institute of Quantum Information and Matter, California Institute of Technology, Pasadena, CA}
\affil[2]{Station Q Quantum Architectures and Computation Group, Microsoft Research, Redmond, WA}







\date{\vspace{-1cm}}
\maketitle 

\begin{abstract}

We give a quantum algorithm for solving semidefinite programs (SDPs). It has worst-case running time $n^{\frac{1}{2}} m^{\frac{1}{2}} s^2 \poly(\log(n), \log(m), R, r, 1/\delta)$, with $n$ and $s$ the dimension and row-sparsity of the input matrices, respectively, $m$ the number of constraints, $\delta$ the accuracy of the solution, and $R, r$ a upper bounds on the size of the optimal primal and dual solutions. This gives a square-root unconditional speed-up over any classical method for solving SDPs both in $n$ and $m$. We prove the algorithm cannot be substantially improved (in terms of $n$ and $m$) giving a $\Omega(n^{\frac{1}{2}}+m^{\frac{1}{2}})$ quantum lower bound for solving semidefinite programs with constant $s, R, r$ and  $\delta$. 


The quantum algorithm is constructed by a combination of quantum Gibbs sampling and the multiplicative weight method. In particular it is based on a classical algorithm of Arora and Kale for approximately solving SDPs. We present a modification of their algorithm to eliminate the need for solving an inner linear program which may be of independent interest.

\end{abstract}


\section{Introduction}

Quantum computers harness the unique features of quantum mechanics to compute in novel ways that outperform classical methods. In the past 20 years a variety of quantum algorithms offering speed-ups over classical computations have been found, including Shor's polynomial-time quantum algorithm for factoring \cite{Shor99} and Grover's quantum algorithm for searching a database in time square-root its size \cite{Gro97}. A central challenge in quantum computing is to identify more quantum algorithms beating classical computing, especially for practically relevant problems.  

Semidefinite programming is one of the most successful algorithmic frameworks of the past few decades \cite{VB96}. It has applications ranging from designing efficient algorithms for approximating combinatorial optimization problems \cite{Goe97} to operations research and beyond \cite{BV04}. The power of semidefinite programs (SDPs) resides in their generality, together with the fact that there are efficient methods for solving them \cite{BV04}. However, given the steadily increasing sizes of SDPs found in practice, it is an important problem to find even more efficient algorithms.  

In this paper we give a quantum algorithm for solving semidefinite programming achieving a quadratic speed-up over any classical method, both in the dimension of the matrices and in the number of constraints of the program (we note however that the algorithm run-time also depends on a parameter measuring the size of the solution of the SDP, discussed below). Our work also gives the first quantum speed-up for linear programming, which is an important subclass of SDPs. We also show that such quadratic speed-up is not far from the best possible in general. We believe the results of this paper make a compelling case that solving SDPs quickly has the potential to be a relevant application of future quantum computers.   



\subsection{Semidefinite Programs}

A general semidefinite program (SDP) is given by
\begin{eqnarray}  \label{SDPprimal}
&& \max \text{tr}(CX) \nonumber \\
\forall j \in [m],  && \text{tr}(A_j X) \leq b_j \nonumber \\
&& X \geq 0,
\end{eqnarray}
where the $n \times n$ Hermitian matrices $(C, A_1, \ldots, A_m)$  and the real numbers $(b_1, \ldots, b_m)$ are the inputs of the problem. The optimization is taken over positive semidefinite $n \times n$ matrices $X$. The dual program of the primal SDP given by Eq.~(\ref{SDPprimal}) is the following:
\begin{eqnarray} \label{dualproblem}
&& \min b .  y   \nonumber \\
&& \sum_{j=1}^m y_j A_j  \geq C\nonumber \\
&& y \geq 0,
\end{eqnarray}
where the minimization is taken over vectors $y := (y_1, \ldots, y_m)$. Under mild conditions the primal and dual problems have the same optimal value \cite{BV04}. We can assume that 
\be
\Vert A_i \Vert \leq 1 \hspace{0.2 cm} \forall i \in [m], \hspace{0.2 cm} \text{and} \hspace{0.2 cm} \Vert C \Vert \leq 1, 
\ee
with $\Vert * \Vert$ the operator norm. This is without loss of generality by normalizing $b_i$ and the optimal solution appropriately. 

There are several different classical polynomial-time algorithms for solving semidefinite programs. One is the class of interior point methods. The state-of-the-art algorithm (in terms of rigorous worst-case bounds) has running time $\tilde{O}(m(m^2 + n^{\omega} + mns) \log(1/\delta))$ \cite{LSW15}, with $\omega$ the exponent of matrix multiplication, $s$ the row-sparsity of the matrices $(C, A_1, \ldots, A_m)$ (i.e., the maximum number of non-zero entries in each of the rows of the matrices), and $\delta$ the accuracy of the solution. 

If one is willing to tolerate a worse scaling with error, faster algorithms can be sometimes obtained using the multiplicative weight method \cite{AHK05, AK07, AHK12}. In particular Arora and Kale gave an algorithm for solving the SDP given by Eq.~(\ref{SDPprimal}) in time $\tilde{O}( nms \left( \frac{Rr}{\delta} \right)^4 + ns \left( \frac{Rr}{\delta} \right)^7 )$ \cite{AK07}, where $R$ and $r$ are upper bounds to the size of the optimal primal and dual solutions. 

It is an important open problem to find even more efficient algorithms for solving SDPs. Nonetheless, as we show in Section \ref{lowerbound}, a limit for any such improvement is the lower bound of $\Omega(n + m)$ for SDPs with $s, \omega, R = O(1)$.

\subsection{Problem Statement} \label{statement}

The problem we want to solve is the following: Given a list of $n \times n$ matrices $(A_1, \ldots, A_m, A_{m+1})$ (with $A_{m+1} := C$), approximate the optimal value of Eqs.~(\ref{SDPprimal}) and (\ref{dualproblem}) and output optimal primal and/or dual solutions. To formulate the problem precisely we must specify in which form the inputs and outputs are given. 

\vspace{0.2 cm}

\noindent \textit{Input Model}: We assume there is an oracle ${\cal P}_A$ that given the indices $j \in [m+1]$, $k \in [n]$ and $l \in [s]$, computes a bit string representation of the $l$'th non-zero element of the $k$-th row of $A_j$,\footnote{We assume all elements of the input matrices can be represented exactly by a bit string of size $\polylog(n, m)$. If not we can truncate the matrices, which will only incur error $\exp(- \polylog(n, m))$ that can be neglected.} i.e. the oracle performs the following map:
\begin{equation}
\ket{j, k, l, z} \rightarrow \ket{j, k, l, z \oplus (A_{j})_{kf_{jk}(l)}}.
\end{equation}
with $f_{jk} : [r] \rightarrow [N]$ a function (parametrized by the matrix index $j$ and the row index $k$) which given $l \in [s]$ computes the column index of the $l$-th non-zero entry.





\vspace{0.2 cm}

\noindent \textit{Output}: One way to specify the output is to require a list with the entries of $X$ or $y$. However it is clear that in such a case at least $n(n-1)/2$ (or $m$) time is required even to write down a primal (or dual) solution. Therefore it is necessary to relax the format of the output in order to obtain faster algorithms. 


We require the quantum algorithm provides the following:

\begin{itemize}
\item An estimate of the optimal objective value.

\item An estimate of $\Vert y \Vert_1$ and/or $\tr(X)$.

\item  \textit{Samples} from the distribution $p := y/\Vert y \Vert_1$ and/or from the quantum state $\rho := X/\tr(X)$.\footnote{In this paper we present a quantum algorithm producing an estimate of $\Vert y \Vert_1$ and samples from the probability distribution $p := y/\Vert y \Vert_1$. The algorithm can be modified to also generate an estimate of $\tr(X)$ and samples from $\rho := X/\tr(X)$. However, we leave the details of this improvement to a future version of the paper.}

\end{itemize}

As we show in Section \ref{lowerbound},  $\Omega(n + m)$  calls to the oracle are required even classically to output a solution as above.

\section{Results}

In this paper we give the first quantum algorithm for solving SDPs offering a speed-up over classical methods. Below we state the main contributions on a high level, describe the algorithm, and present a few open questions related to it. 


\subsection{Main Ideas}

The first contribution of the paper is to notice that classical algorithms for solving SDPs based on the multiplicative weight method \cite{AK07} imply that in order to solve the SDP of Eq.~(\ref{SDPprimal}), it is enough to prepare quantum thermal (Gibbs) states of Hamiltonians given by linear combinations of the input matrices $A_1, \ldots, A_m, C$ of the program. Therefore in cases where such Gibbs states can be prepared efficiently (in time polynomial in $\log(n)$), quantum computers can give exponential speed-ups. This already suggests that our method might be an interesting heuristic to run on quantum computers (e.g., by using quantum Metropolis sampling \cite{Tem11, YAG12} to prepare the Gibbs states). 

The second contribution is to combine the first observation with amplitude amplification \cite{Bra02} in the preparation of the Gibbs state to achieve a generic quadratic speed-up in terms of $n$, the dimension of the input matrices of the program. Here we can apply known results \cite{PW09, CS16} on using amplitude amplification to prepare Gibbs states on a quantum computer in time given roughly by the square root of the dimension of the system. 

The third contribution is to show one can achieve a quadratic speed-up also in $m$, the number of constraints of the program. Establishing this fact requires more work. We modify the Arora-Kale algorithm and replace the inner linear program they use by the preparation of a Gibbs state of a classical Hamiltonian, whose entries are estimated from the expectation value (with each of the input matrices) of the Gibbs state prepared in the main thread of the algorithm. We show this replacement is possible using Jaynes' principle of maximum entropy \cite{Jay57}, or more specifically a recent approximate version of it \cite{LRS15} with better control of parameters. Then we apply amplitude amplification also to the preparation of this Gibbs state. This modification alone is not enough to give a speed-up, since the sparsity of the Hamiltonians for which we must prepare the associated Gibbs states also depends on $m$ (and the quantum algorithms for preparing Gibbs states we consider have linear dependence on sparsity). We then show that we can sparsify the Hamiltonians under consideration by random sampling, without changing the functioning of the algorithm, so that their sparsity only depends on the original sparsity $s$ of the input matrices (up to polylogarithmic factors in $n, m$).




\subsection{The Algorithm}

\noindent \textbf{Reduction to Feasibility:}  Using binary search we can reduce the optimization problem to a feasibility one. Let $\alpha$ be a guess for the optimal solution (which will be varied by binary search). We are then concerned with the problem of either sampling from a probability distribution $p := y/\Vert y \Vert_1$, with $y$ a dual feasible vector whose value is at most $\alpha(1+\delta)$ for a small $\delta > 0$, or finding out that the optimal value is larger than $\alpha(1 - \delta)$.\footnote{Alternatively we could also sample from a quantum state $\rho := X/\tr(X)$, with $X$ a primal feasible solution with objective value at least $\alpha(1 - \delta)$; however we do not consider this task in the current version.}

\vspace{0.2 cm}

\noindent \textbf{Gibbs Samplers:} A subroutine of the main algorithm is the following:

\begin{dfn}
[Gibbs Sampler] Let $H$ be a Hamiltonian and $O[H]$ an oracle for its entries.\footnote{In analogy with the input model 1, $O[H]$ is an oracle that given the indices $k \in [n]$ and $l \in [s]$, computes a bit string representation of the $l$'th non-zero element of the $k$-th row of $H$. Here $n$ is the dimension of $H$ and $s$ its sparsity.} Then $\GS(O[H], \nu)$ is a quantum operation that given access to $O[H]$, outputs a state $\rho$ such that $\Vert \rho - e^{H}/\tr(e^H) \Vert_1 \leq \nu$.
\end{dfn}

Several different Gibbs samplers have been proposed in the literature \cite{Tem11, YAG12, PW09, CS16, KB16, BK16}, and any of them could be used in the main quantum algorithm.  

\vspace{0.2 cm}

\noindent \textbf{Oracle by Estimation:} Given a Hamiltonian $H$, in order to run a Gibbs sampler algorithm, we need an oracle for its entries. We also need the notion of a probabilistic oracle. This is an oracle that with high probability outputs the right entry of the Hamiltonian, but with small probability might output a wrong value.  

Consider a quantum state $\rho$ and two real numbers $\lambda$ and $\mu$. We define the Hamiltonian $h(\rho, \lambda, \mu) := \sum_{i=1}^m  r_i \ket{i}\bra{i}$, with 
\be
r_i := \lambda \tr(A_i \rho) + \mu b_i, 
\ee
and its truncated version
\be \label{hbar}
\overline{h}(\rho, \lambda, \mu) := \sum_{i=1}^m  \overline{r}_i \ket{i}\bra{i},
\ee
with $\overline{r}_i$ the rounding of $r_i$ to precision $h_{\text{precision}}$. Throughout the paper we set 
\be
h_{\text{precision}} = \frac{\delta}{56 R^2}.
\ee
The quantum algorithm will make use of calls to a probabilistic oracle for $\overline{h}$, which we show how to construct in Section \ref{oracleforhbar}. A subtlety is that $\rho$ will not be given explicitly as a density matrix, but only as a quantum state. Therefore in order to implement the oracle for $\overline{h}$, we need to first estimate (some of) the values $\{ \tr(A_i \rho) \}_{i=1}^m$.




\vspace{0.2 cm}

\noindent \textbf{The Size Parameter $R$:} Apart from the dimension of the matrices $n$, the number of constraints  $m$, the sparsity of the input matrices $s$, and the error $\delta$, the algorithm will depend on another parameter of the SDP. For many problems of interest, this is a constant independent of $n$ and $m$.\footnote{We remind the reader we are also assuming that $\Vert C \Vert, \Vert A_i \Vert \leq 1$ for all $i \in [m]$ (which is w.l.o.g. by changing the values of the numbers $b_j, \alpha$ appropriately).} 

Following \cite{AK07}, we assume $A_1 = I$ and let $b_1 = R$. Thus we have the constraint 
\begin{equation}
\tr(X) \leq R.
\end{equation}

We can always add this constraint without changing the optimal solution by choosing $R$ sufficiently large. The parameter $R$ is a measure of the size of the optimal solution of the SDP. Note we have the upper bound $\alpha \leq R$.\footnote{It follows from $\tr(C X) \leq \Vert X \Vert_1 \Vert C \Vert \leq \tr(X) \leq R$.} We also assume $R$ is chosen sufficiently large such that $\max_i |b_i| = R$.


\vspace{0.2 cm}

\noindent \textbf{Reduction to $b_i \geq 1$ and $\alpha \geq 1$ and the dual size parameter $r$:} Our algorithm will only work for SDPs for which $b_i \geq 1$ for all $i \in [m]$. However in Appendix \ref{appendix} we prove Lemma \ref{reductionpositiveb} below, which shows that if we can solve SDPs with the $b_i$'s larger than one, we can solve arbitrary SDPs in roughly the same time. We say a feasible solution is  $\delta$-optimal if its objective value is within additive error $\delta$ to the optimal. Given the SDP of Eq. (\ref{dualproblem}) with an optimal solution $(y_1, \ldots, y_m)$, we assume we are given an upper bound $r$ to the sum $\sum_{i=1}^m y_i$ (if there is more than one solution, any of them can be chosen). We have:

\begin{lem} \label{reductionpositiveb}
One can sample from a $\delta$-optimal solution of the SDP given by Eq. (\ref{dualproblem}) (with dimension $n$, $m$ variables, size parameter $R$ and upper bound $r$ on optimal solution vector) given the ability to sample from a $(\delta/r)$-optimal solution of the SDP given by Eq.~(\ref{dualproblem}) (with dimension $n+1$, $m+1$ variables and size parameter $2R+1$) in which $b_i \geq 1$ for all $i \in [m]$.
\end{lem}

To apply Lemma \ref{reductionpositiveb} we need an upper bound $r$ on the $\ell_1$-norm of the optimal solution vector. For example, if all $b_i$'s are positive, we can take $r = \alpha / \min_{i}b_i$. 



We will also assume that $\alpha \geq 1$. We can ensure this is the case as follows: Given the SDP of Eq.~(\ref{dualproblem}) with all $b_i$'s larger than one, it is clear that we can take $\alpha > 0$ (as all the $y_i$'s are non-negative). Suppose $\alpha < 1$. Then we consider a new SDP with the $b_i'$s rescaled by $1/\alpha$. As an effect we must solve the scaled SDP to accuracy $\delta \alpha$. Note that the complexity of solving the SDP (which depends on the accuracy) increases when $\alpha$ approaches zero. It is an open question this drawback can be avoided.  

The quantum algorithm for $\alpha \geq 1$, $b_i \geq 1$ discussed below is given in terms of multiplicative error, while the reduction above works with additive error. It is easy to convert the multiplicative approximation to additive approximation by redefining the error $\delta \rightarrow \delta/\alpha$. This results in an increased complexity of the algorithm in terms of $\alpha$. 

\vspace{0.2 cm}


\noindent \textbf{Probability of success:} The probability of success of our algorithm is greater than
\be
1 - O(\exp(-log^{\xi}(nm))), 
\ee
where $\xi>0$ is a free parameter.

\vspace{0.2 cm}

\noindent \textbf{The Algorithm:}  Let 
\be
\gamma := \left  \lceil \frac{8}{\varepsilon^2} \log(m) R^2 \right \rceil
\ee
for an $\varepsilon>0$ defined below. For an integer $k \leq \gamma$, we define
\begin{equation} 
\overline{h}(\rho, k) := \overline{h}\left(\rho, - \frac{\varepsilon}{8R^2}k, - \frac{\varepsilon}{8R^2}(\gamma - k) \right),
\end{equation}
with $\overline{h}$ the Hamiltonian given by Eq.~(\ref{hbar}). Let $e_1 := (1, 0, \ldots, 0)$ be the first computational basis state of $\mathbb{R}^m$. Let
\be
G_{\overline{h}}(\rho) := \max_{k \in [\gamma]} \hspace{0.1 cm} \left( \text{number of calls to the oracle}  \hspace{0.1 cm}  O[\overline{h}(\rho, k)]  \hspace{0.1 cm} \text{in}  \hspace{0.1 cm}  \GS \left(O[\overline{h}(\rho, k)], \frac{\varepsilon}{4} \right) \right).
\ee

The main algorithm is the following:

\mbox

\begin{mybox}
\begin{algorithm}[Quantum Algorithm for SDPs]
\mbox{}\label{alg:q}
  \begin{description}
  \item[Input:]  Oracles for $\{ A_1, \ldots, A_m, C \}$, with $\Vert C \Vert$, $\Vert A_i \Vert \leq 1$ and $\{ b_1, \ldots, b_m \}$ with $b_i \geq 1$.  Parameters $R, \alpha, \delta, \xi >0$.
  \item[Output:] Either a sample from distribution $p$ and a real number number $L$ such $y := L p$ is dual feasible with objective value less than $(1 + \delta)\alpha$, or the label $\mathsf{Larger}$ indicating the optimal objective value is larger than $(1 - \delta)\alpha$.  
    
  \end{description}

\vspace{0.2 cm}

Set $\rho^{(1)} = I/n$. Let $\varepsilon = \frac{\delta}{28 R^2}$,  $\varepsilon' = - \ln(1 - \varepsilon)$, $M = 80 \log^{1+ \xi}(8 R^2 n m/\varepsilon)/\varepsilon^2$,  $L = 80 \log^{1 + \xi}(nm)/\varepsilon^2$ and $Q = 10^6   R^{6} \ln^{2 + \xi}(nm)/ \delta^4$. Let $T = \frac{500 R^3 \ln(n)}{\delta^2}$. For $t = 1, \ldots, T$:
\begin{enumerate}

\item Set $y^{(t)} = (0, \ldots, 0)$.

\item For $k = 1, \ldots, \gamma = \left  \lceil \frac{8}{\varepsilon^2} \log(m) R^2 \right \rceil$, $N = 1, \ldots, \lceil \frac{\alpha}{\varepsilon} \rceil$, 
\begin{itemize}

\item Create $M$ copies of $q \leftarrow \GS(O[\overline{h}(\rho^{(t)}, k)], \varepsilon/4)$.

\item Sample $i_1, \ldots, i_M$ independently from the distribution $q$. 

\item Compute estimates $\{ e_{i_1}, \ldots, e_{i_M}, f \}$ of $\{ \tr(A_{i_1} \rho^{(t)}), \ldots,  \tr(A_{i_M} \rho^{(t)}), \tr(C \rho^{(t)})  \}$ to accuracy $\varepsilon/2$ using $(M+1)L$ samples from $\rho^{(t)}$ .

\item If $1/M \sum_{j=1}^M  e_{i_j} \geq f / (\varepsilon N) - \varepsilon$ and $1/M \sum_{j=1}^M  b_{i_j} \leq \alpha / (\varepsilon N) + R \varepsilon$, set $k_t = k$, $N_t = N$, $q^{(t)} = q$  and $y^{(t)} = \varepsilon N q^{(t)}$.

\end{itemize}
 
\item If $y^{(t)} = (0, \ldots, 0)$, stop and output $\mathsf{Larger}$.

\item Create $Q+1$ copies of $q^{(t)} \leftarrow \GS(O[\overline{h}(\rho^{(t)}, k_t)], \varepsilon/4)$

\item Sample $i_1, \ldots, i_Q$ independently from the distribution $q^{(t)}$.

\item Let $M^{(t)} = \left(    \varepsilon N_t   Q^{-1} \sum_{j=1}^Q A_{i_j} - C + 2 \alpha I \right)/ 4 \alpha$.

\item Let $C_{t} := \frac{10 \log(m)}{\varepsilon^2}(\frac{\gamma \alpha}{\varepsilon} M  + Q )G_{\overline{h}}(\rho^{(t)}) +  \frac{2\gamma \alpha}{\varepsilon}ML$. \\ Create $C_t$ copies of the state $\rho^{(t+1)} \leftarrow \GS(- \varepsilon' (\sum_{\tau = 1}^t M^{(\tau)}), \varepsilon/4)$. 

\end{enumerate}

Output $ \Vert \overline{y}  \Vert_1$ and a sample from $\overline{y}/ \Vert \overline{y}  \Vert_1$ with $\overline{y} =  \frac{\delta \alpha}{2R} e_1 + \frac{1}{T} \sum_{t=1}^T y^{(t)}$.
\end{algorithm}
\end{mybox}

\mbox{}

\vspace{0.1 cm}


\vspace{0.2 cm}

Let 
\be
G_{M} := \max_{t \leq T} \left(  \text{number of calls to the oracle in} \hspace{0.1 cm} \GS \left( O \left [- \varepsilon'  \left( \sum_{\tau = 1}^t M^{(\tau)} \right) \right], \frac{\varepsilon}{4} \right) \right),
\ee
and
\be
G_{\overline{h}} := \max_{t \leq T}  G_{\overline{h}}(\rho^{(t)}). 
\ee

Finally, let $T_{\text{Meas}}$ be the maximum time needed to estimate one of $\tr(A_i \rho)$, for $i \in [m]$, and  $\tr(C \rho)$ (for an arbitrary $\rho$) within additive error $\varepsilon/2$ (in Lemma \ref{estimatingAi} we show that for $s$-sparse matrices, $T_{\text{Meas}} \leq \tilde{O}(s/\varepsilon^2)$).



We prove in Section \ref{proofmain} the following:

\begin{thm} \label{mainthm}
Algorithm \ref{alg:q} runs in time 
\be
\tilde{O} \left( \frac{(R^{21}}{\delta^{11}} G_{\overline{h}} G_M  \right) + \tilde{O} \left( \frac{R^{13}}{\delta^{5}} T_{\text{Meas}}  \right).
\ee
The algorithm fails with probability at most 
\be
O((R/\delta)^{18} (nm)^{10}\exp(- \log^{\xi}(nm))).
\ee
Assuming it does not fail, if it outputs $\mathsf{Larger}$, the optimal objective value is larger than $(1 - \delta)\alpha$. Otherwise it outputs a sample of a probability distribution $p$ and a real number $L$ such that $y = L p$ is dual feasible and $ \sum_i y_i b_i \leq (1 + \delta) \alpha$.
\end{thm}

We note the algorithm is very costly in terms of the size parameter $R$ and the error $\delta$. We believe it is possible to significantly reduce the complexity in terms of these two parameters, but we leave this possibility as an open question to future work. 

One interesting Gibbs sampler to consider is quantum Metropolis \cite{Tem11, YAG12}. Although it is difficult to obtain rigorous estimates on its running time, it is expected that it is polylogarithmic in many cases. Whenever this is the case for the Hamiltonians involved in the algorithm (given by linear combinations of the $A_i$'s and $C$), one would achieve exponential speed-ups. 

In Section \ref{proofmain} we show:

\begin{cor}  \label{cormain}
Using the Gibbs Sampler from Ref. \cite{PW09}, Algorithm \ref{alg:q} runs in time $\tilde{O}( n^{\frac{1}{2}} m^{\frac{1}{2}} s^2 R^{32}  / \delta^{18})$.
\end{cor}

As we show in the next section, this represents an unconditional polynomial speed-up (in terms of $m$ and $n$) over any classical method for solving semidefinite programming.

One particular case of interest is when $R$ is a constant independent of all other parameters. In this case the running time only depends on the parameters $n, m, s, \delta$. Moreover the SDP has a clear quantum interpretation: we want to optimize the expectation value of an observable $C/R$ on a quantum state $\rho$ subject to the constraints that the expectation value of $A_i$ on $\rho$ is bounded by $b_i/R$, for all $i \in [m]$.

\subsection{Lower Bounds}   \label{lowerbound}

We give a lower bound on the complexity of solving SDPs which shows the $n, m$ dependence of Algorithm \ref{alg:q} cannot be substantially improved. Consider the following two instances of the primal problem given by Eq.~(\ref{SDPprimal}), with $R = 1$, $A_1 = I$, $b_j = 1$ for $j \in [m]$ and either 

\begin{enumerate}

\item For a random $i \in [n]$, set $C_{ii} = 1$. All other elements of $C$ are set to zero. Choose at random $j \in [m]$ and set $(A_{j})_{ii} = 2$. All other elements of the matrices $\{ A_j \}_{i=2}^m$ are set to zero.

\item For a random $i \in [n]$, set $C_{ii} = 1$. All other elements of $C$ are set to zero. All elements of the matrices $\{ A_j \}_{i=2}^m$ are set to zero.
\end{enumerate}

We claim that to decide which of the two cases we are given requires at least $\Omega(n+m)$ calls to the oracle classically and $\Omega(\sqrt{n}+\sqrt{m})$ calls quantum-mechanically. This follows from an elementary reduction to the search problem. 

It is easy to see that the optimal solution of the primal and dual problems in the first case are $X = \ket{i}\bra{i}/2$ and $y = \ket{j}\bra{j}$, with objective value $1/2$.  In the second case, in turn, $X = \ket{i}\bra{i}$ and $y = \ket{1}\bra{1}$, with objective value 1. Therefore we can decide which of the two we are given and find the marked $(i, j)$ (in the first case) given samples of the optimal $y$ and $X$ and a constant-error approximation to $\tr(X)$ or $\Vert y \Vert_1$. This is equivalent to solving two search problems, one in a list of $n$ elements and another in a list of $m$ elements.

\subsection{Discussion and Open Questions}

The core quantum part of the algorithm is the preparation of quantum Gibbs states. Classically there are several interesting applications of the Monte Carlo method and the Metropolis algorithm to problems not related to simulating thermal properties of physical systems \cite{Gil05}. One could expect the same will be the case for quantum Metropolis. We might have to wait until there are working quantum computers to fully explore the usefulness of quantum Metropolis, since in analogy to the classical case, many times heuristic methods based on it might work well in practice even though it is hard to get theoretical guarantees. Nevertheless, as far as we know the results of this paper give the first example of a problem of interest outside the simulation of physical systems in which "quantum Monte Carlo" methods (i.e., sampling from quantum Gibbs states) play an important role. The algorithm can also be seen as a new application of quantum annealing. One difference is that in this case the annealing is used to prepare a finite temperature state, instead of a groundstate as is usually considered in quantum adiabatic optimization. 

The algorithm is also inherently robust, in the sense that to compute a solution of the SDP to accuracy $\varepsilon$, it suffices to be able to prepare  approximations to accuracy $O(\varepsilon)$ of the Gibbs state of Hamiltonians given by linear combinations of the input matrices. Moreover we believe the constants and the dependence on $R$ and $\delta$ might be substantially improved by a more careful analysis. If this turns out to be indeed the case, we expect the algorithm to be a promising candidate for a relevant application of small quantum computers (even without the need for error correction).  

This work leaves several open questions for future work. For example:

\begin{itemize}

\item The algorithm has very poor scaling in terms of $R$ and $\delta$. It is a pressing open question to improve its running time in terms of these parameters. Also can we close the gap in terms of $n$ and $m$ between the lower bound ($\Omega(\sqrt{n} + \sqrt{m})$) and the algorithm ($O(\sqrt{nm})$)?

\item Although quadratic speed-ups in terms of $n$ and $m$ are the best possible in the worst case, it is an interesting question whether more significant speed-ups are possible in specific instances. How large are the speed-ups on average (for example choosing the input matrices at random from a given distribution)? Even more interesting is to explore whether there is a SDP of practical interest for which we might have larger quantum speed-ups. 


\item How robust is the algorithm to noise? Can we run it without the need of quantum error correction in analogy to what has been proposed for quantum annealing? Is there an improvement of the algorithm which would be suitable for a small-scale quantum computer (with hundreds of physical qubits)?

\item The multiplicative method is an important algorithmic technique classically. In this paper we give an application of the matrix multiplicative weight method to quantum algorithms. Are there more applications? 

\item Can we enlarge the class of optimization problems having a quantum speed-up beyond SDPs? In particular, can we get quantum speed-ups for optimizing general convex functions over convex sets (assuming we have an efficient oracle for membership in the set)? 

\item In practice the preferred algorithms for solving SDPs are based on the interior point method. Can we also find a quantum algorithm for SDPs based on it? 

\end{itemize}


\section{Analysis of the Quantum Algorithm}

\subsection{The Arora-Kale Algorithm}

The quantum algorithm builds on a classical algorithm of Arora and Kale for solving SDPs, which we now review. One element of their approach is an auxiliary algorithm termed $\ORACLE(\rho)$, which given a density matrix $\rho$, searches for a vector $y$ from the polytope 
\begin{equation} \label{polytopeDalpha}
{\cal D}_{\alpha} := \{ y \in {\mathbb{R}^m} : y \geq 0, b . y \leq \alpha  \} 
\end{equation}
such that
\begin{equation} \label{kindoffeasible}
\sum_{j=1}^m y_j  \tr(A_j \rho)  \geq \tr(C \rho),
\end{equation}
or outputs \textit{fail} if no such vector exists. 

The running time of their algorithm also depends on the so-called width $\omega$ of the SDP, defined as
\begin{equation}
\omega := \max_{y \in {\cal D}_{\alpha}} \left \Vert \sum_j y_j A_j - C \right \Vert.
\end{equation}
We note the bound:\footnote{Assuming $b_i \geq 1$.}
\begin{eqnarray} \label{boundbyalpha}
\omega &=& \max_{y \in {\cal D}_{\alpha}} \left \Vert \sum_j y_j A_j - C \right \Vert \nonumber \\ &\leq& \max_{y \in {\cal D}_{\alpha}} \sum_j y_j + 1   \nonumber \\ 
&\leq&  \max_{y \in {\cal D}_{\alpha}} \sum_j b_j y_j + 1  \nonumber \\ &\leq& \alpha + 1 \nonumber \\
&\leq& R + 1. 
\end{eqnarray}

The Arora-Kale algorithm is the following:

 \mbox

\begin{mybox}
\begin{algorithm}[Arora-Kale Algorithm for SDPs]
\mbox{}\label{alg:basic}
    %

\vspace{0.2 cm}

Set $\rho^{(1)} = I/n$. Let $\varepsilon = \frac{\delta \alpha}{2 R^2}$, and let $\varepsilon' = \ln(1 - \varepsilon)$. Let $T \geq \frac{16 R^4 \ln(n)}{\alpha^2 \delta^2}$. For $t = 1, ..., T$:

\begin{enumerate}

\item Run $\ORACLE(\rho^{(t)})$. If it fails, stop and output $\rho^{(t)}$. 

\item Else, let $y^{(t)}$ be the vector generated by $\ORACLE(\rho^{(t)})$.

\item Let $M^{(t)} = (\sum_{j=1}^m A_j y_j^{(t)} - C + \omega I) / 2 \omega$

\item Compute $W^{(t+1)} = \exp( - \varepsilon' (\sum_{\tau = 1}^t M^{(\tau)}))$.

\item Set $\rho^{(t+1)} = \frac{W^{(t+1)}}{\tr(W^{(t+1)})}$ and continue.

\end{enumerate}
\end{algorithm}
\end{mybox}

\mbox{}

The central idea behind the algorithm is a variant of the multiplicative weight method for positive semidefinite matrices. Let us denote by $\lambda_n(X)$ the minimum eigenvalue of the $n \times n$ Hermitian matrix $X$. Arora and Kale proved the following:

\begin{lem} \label{MMWM}
[Matrix Multiplicative Weights method; Theorem 10 of \cite{AK07}] Let $M^{(t)}$ be such that $0 \leq M^{(t)} \leq I$ for every $t$. Fix $\varepsilon < \frac{1}{2}$, and let $\varepsilon' =  - \ln(1 - \varepsilon)$. define $W^{(t)} = \exp \left(  - \varepsilon'  \left(  \sum_{\tau = 1}^{t-1} M^{(\tau)}  \right) \right)$ and the density matrices $\rho^{(t)} = \frac{W^{(t)}}{\tr(W^{(t)})}$. Then
\be
\sum_{t=1}^T  \tr(M^{(t)} \rho^{(t)}) \leq (1 + \varepsilon) \lambda_{n}\left( \sum_{t=1}^T M^{(t)}  \right) + \frac{\ln(n)}{\varepsilon}.
\ee
\end{lem}

Let $e_1 := (1, 0, \ldots, 0)$. Then the main result of \cite{AK07} is the following (as a warm up to the proof of correctness of the quantum algorithm we reproduce the argument of Arora and Kale below): 

\begin{thm} [Theorem 1 of \cite{AK07}] Suppose $\ORACLE$ never fails for $ T =  \frac{16 R^4 \ln(n)}{\alpha^2 \delta^2}$ iterations. Then $\overline{y} = \frac{\delta \alpha}{R} e_1 + \frac{1}{T} \sum_{t=1}^T y^{(t)}$ is dual feasible with objective value at most $\alpha(1 + \delta)$ . 
\end{thm}

\begin{proof}

By the definition of $\ORACLE$ and the fact that $A_1 = I$, $b_1 = R$, we have 
\begin{equation}
\overline{y}. b = \delta\alpha + \frac{1}{T}  \sum_{t=1}^T y^{(t)}.b \leq \alpha(1 + \delta).
\end{equation}

So it remains to show that $\overline{y}$ is dual feasible. Again by the properties of $\ORACLE$, $\overline{y} \geq 0$. We now use Lemma \ref{MMWM} to show $ \sum_{j=1}^m \overline{y}_j A_j  \geq C$. Indeed
\begin{eqnarray}   \label{mmwpsd}
 && \lambda_{n}\left( \sum_{j=1}^m \overline{y}_j A_j - C \right) \\ &=&  \lambda_{n}\left(  \frac{1}{T} \sum_{t=1}^T \sum_{j=1}^m y_j^t  A_j - C \right) + \frac{\delta \alpha}{R}  \\
 &=& 2 \omega \lambda_{n}\left(  \frac{1}{T} \sum_{t=1}^T \left(  \sum_{j=1}^m y_j^t  A_j - C + \omega I \right)/2 \omega    \right) - \omega + \frac{\delta \alpha}{R}   \\
 &\overset{(i)}{\geq} &  \frac{2 \omega}{(1 + \varepsilon)} \frac{1}{T} \sum_{t=1}^T  \text{tr}\left( \rho^{(t)} \left(  \sum_{j=1}^m y_j^t  A_j - C + \omega I \right)/2 \omega    \right) - \frac{2 \omega \ln(n)}{T(1 + \varepsilon)\varepsilon} - \omega + \frac{\delta \alpha}{R}\label{beforebeforelast} \\
&\overset{(ii)}{\geq}& \frac{\omega}{(1 + \varepsilon)}   - \frac{2 \omega \ln(n)}{T(1 + \varepsilon)\varepsilon} - \omega + \frac{\delta \alpha}{R}  \label{beforelast} \\
 &=& - \frac{2 \omega \ln(n)}{T(1 + \varepsilon)\varepsilon} - \frac{\varepsilon \omega}{(1 + \varepsilon)} + \frac{\delta \alpha}{R}   \\
&\geq&   - \frac{2 \omega \ln(n)}{T(1 + \varepsilon)\varepsilon}  + \frac{\delta \alpha}{2R}      \\
&\overset{(iii)}{\geq}& 0. \label{last}
\end{eqnarray}

Inequality (iii)  follows from Eq.~(\ref{boundbyalpha}) and the choices of $T$ and $\varepsilon$. Inequality (ii) follows from Eq.~(\ref{kindoffeasible}) in the definition of $\ORACLE$ which gives
\begin{equation} 
\frac{1}{T} \sum_{t=1}^T  \text{tr}\left( \rho^{(t)} \left(  \sum_{j=1}^m y_j^t  A_j - C + \omega I \right)/2 \omega    \right) \geq \frac{1}{2}.
\end{equation}
Finally inequality (i) follows from the Matrix Multiplicative Weight method (Lemma \ref{MMWM}), which can be applied since by the definition of $\omega$, 
\begin{equation}
 0 \leq \left(\sum_{j=1}^m y_j^t  A_j - C + \omega I \right)/2 \omega \leq I.
\end{equation}
\end{proof}

\subsection{Approximately Implementing $\ORACLE$ by Gibbs Sampling}   \label{oraclesection}

As a step towards the quantum algorithm, we now give an explicit implementation of the $\ORACLE$ auxiliary algorithm. The idea is to use the fact that by Jaynes' principle, we can w.l.o.g. take the output $y$ of $\ORACLE(\rho)$ to be (up to normalization) a Gibbs probability distribution over the two constraints, i.e., a distribution over $[m]$ of the form 
\begin{equation}
q_{\rho, \lambda, \nu}(i) := \frac{\exp \left( \lambda \tr(A_i \rho) + \mu \sum_i b_i  \right)}{\sum_i \exp \left( \lambda \tr(A_i \rho)  + \mu \sum_i b_i  \right)},
\end{equation}
for real numbers $\lambda, \mu$. 

The following is a special case of Lemma 4.6 of \cite{LRS15} (obtained by taking the reference state to be maximally mixed) and is an approximate version of Jaynes' principle with a quantitative control of the parameters of the Hamiltonian in the Gibbs state (in contrast, in the original Jaynes' principle there is no control over the size of the interaction strengths of the Hamiltonian)

Let ${\cal M}(\mathbb{C}^n)$ and ${\cal D}(\mathbb{C}^n)$ be the set of Hermitian and density matrices over $\mathbb{C}^n$. Let ${\cal T} \subseteq {\cal M}(\mathbb{C}^n)$ be a compact set of matrices. We set 
\be
\Delta({\cal T}) := \sup_{A \in {\cal T}} \Vert  A \Vert. 
\ee
For $A \in {\cal M}(\mathbb{C}^n)$, we define the associated dual norm
\begin{equation}
[A]_{\cal T} := \sup_{B \in {\cal T}} \text{tr}(BA). 
\end{equation}
 
\begin{lem}   \label{finitaryJaynes}
(Lemma 4.6 of \cite{LRS15}) For every $\kappa > 0$, the following holds. Let ${\cal T} \subseteq {\cal M}(\mathbb{C}^m)$ be a compact set of matrices and let $\pi \in {\cal D}(\mathbb{C}^m)$ be a density matrix. If one defines $\gamma = \lceil \frac{8}{\kappa^2} \log(m) \Delta({\cal T})^2  \rceil$ then there exist $X_1, \ldots, X_{\gamma} \in {\cal T}$ such that
\begin{equation}
\tilde{\pi} := \frac{  \exp \left( - \frac{\kappa}{4 \Delta({\cal T})^2} \sum_{i=1}^{\gamma} X_i \right)   }   {\tr \left(  \exp \left( - \frac{\kappa}{4 \Delta({\cal T})^2} \sum_{i=1}^{\gamma} X_i \right)  \right)}
\end{equation}
satisfies
\begin{equation}
[\pi - \tilde{\pi}]_{\cal T} \leq \kappa. 
\end{equation}   
\end{lem}

The finitary version of Jaynes' principle above allows us to implement $\ORACLE$ assuming we have access to samples from the distributions $q_{\rho, \lambda,  \mu}$. In fact, for the quantum algorithm it will
be useful to prove a generalization in which we only assume we can sample from distributions $\overline{q}_{\rho, \lambda,  \mu}$ \textit{close} in variational distance to $q_{\rho, \lambda,  \mu}$.

For an integer $k$, let
\begin{equation}
q_{\rho, k} := q_{\rho, - \frac{\kappa}{4R^2} k,  - \frac{\kappa}{4R^2}(\gamma - k)},
\end{equation}
with
\begin{equation}
\gamma := \left  \lceil \frac{8}{\kappa^2} \log(m) R^2 \right \rceil.
\end{equation}

\mbox

\begin{mybox}
\begin{algorithm}[Instantiation of $\ORACLE(\rho)$ by Sampling]
\label{alg:implementingoracle}
\mbox{}

  \begin{description}
  \item[Input:]   Samples from distributions $\overline{q}_{\rho, k}$ such that $\Vert \overline{q}_{\rho, k} - q_{\rho, k} \Vert_1 \leq \nu$ and real numbers $\{ e_{i} \}_{i=1}^{m+1}$ such that $|e_i - \tr(A_i \rho)| \leq \nu$. A parameter $\kappa > 0$. 
  \item[Output:] Samples from distribution $y/ \Vert y \Vert_1$ and value of $\Vert y \Vert_1$ satisfying Eqs. (\ref{firsteqfory}) and (\ref  {secondeqfory}). 
    
  \end{description}

\vspace{0.2 cm}

\item Let $M = 80 \log^{1+ \xi}(8 R^2 n m/\varepsilon)/\varepsilon^2$. For $k = 1, \ldots, \gamma$ and $N = 1, \ldots, \left \lceil \frac{\alpha}{\kappa}  \right \rceil$:
\begin{itemize}

\item Sample $i_1, \ldots, i_M \in [m]^{M}$ independently from the distribution $\overline{q}_{\rho, k}$. 

\item If $\frac{1}{M} \sum_{j=1}^M  e_{i_j} \geq \frac{e_{m+1}}{\kappa N} - (\kappa + \nu)$ and $\frac{1}{M} \sum_{j=1}^M  b_{i_j} \leq \frac{\alpha}{\kappa N} + R(\kappa + \nu)$, output samples from $\tilde{q}_{\rho, k}$ and the number $\kappa N$ (as $y/ \Vert y \Vert_1$ and $\Vert y \Vert_1$, respectively).

\end{itemize}

\end{algorithm}
\end{mybox}

\mbox{}

\vspace{0.1 cm}


\begin{lem}  \label{algorthmfororacleworks}
Suppose $\ORACLE(\rho)$ does not fail. Then with probability larger than $1 - \exp(-\log^\xi(n m))$, Algorithm \ref{alg:implementingoracle} outputs $y$ such that
\begin{equation} \label{firsteqfory}
\sum_{i=1}^m b_i y_i \leq \alpha(1 + 2R(\kappa + \nu)),
\end{equation}
and
\begin{equation} \label{secondeqfory}
\sum_{i=1}^m y_i \tr(A_i \rho) \geq \tr(C \rho) - 2  \alpha (\kappa + \nu).
\end{equation}
\end{lem}
\begin{proof}


By the Chernoff bound and the union bound over all all $k \in [\gamma]$, with probability at least $1 - \exp(-\log^\xi(n m))$, sampling $i_1, \ldots, i_M$ independently from $\overline{q}_{\rho, k}$ guarantees that
\begin{equation} \label{chernoffconstriants}
\left | \sum_{i=1}^m \overline{q}_{\rho, k}(i)  \tr(A_{i} \rho) -   \frac{1}{M} \sum_{j=1}^M e_{i_j}  \right | \leq  \kappa + \nu,
\end{equation}
and
\begin{equation} \label{chernoffconstriant2}
\left | \sum_{i=1}^m \overline{q}_{\rho, k}(i)  b_i  - \frac{1}{M} \sum_{j=1}^M b_{i_j} \right | \leq  \kappa + \nu,
\end{equation}
for all $k \leq \gamma$. 

Let $k \leq \gamma, N \leq  \lceil \frac{R}{\kappa}  \rceil$ be the smallest integers (assuming they exist) such that 
\begin{equation} \label{eq1auxlemmajaynes}
\sum_{i=1}^m \overline{q}_{\rho, k}(i) \tr(A_{i} \rho) \geq  \frac{\tr(C \rho)}{\kappa (N+1)} - (\kappa + \nu),
\end{equation}
and
\begin{equation}  \label{eq2auxlemmajaynes}
\sum_{i=1}^m \overline{q}_{\rho, k}(i) b_i \leq  \frac{\alpha}{\kappa N} + R(\kappa + \nu).
\end{equation}
Then by Eqs.~(\ref{chernoffconstriants}) and (\ref{chernoffconstriant2}), with probability larger than $1 - \exp(-\log^\xi(n m))$ over the choice of $i_1, \ldots, i_M$, 
\begin{equation}
\frac{1}{M} \sum_{j=1}^M  e_{i_j}  \geq \frac{\tr(C \rho)}{\kappa (N+1)} - 2 (\kappa + \nu),
\end{equation}
and 
\begin{equation}
\frac{1}{M} \sum_{j=1}^M  b_{i_j} \leq \frac{\alpha}{\kappa N} + 2R(\kappa + \nu),
\end{equation}
where we used $\max_i |b_i| = R$. The algorithm will then output $y = \kappa N \tilde{q}_{\rho, k}$, which satisfies Eqs.~(\ref{firsteqfory}) and (\ref{secondeqfory}). 

It remains to prove the existence of at least one pair $k \leq \gamma, N \leq  \left \lceil \frac{\alpha}{\kappa}  \right \rceil$ satisfying Eqs.~(\ref{eq1auxlemmajaynes}) and (\ref{eq2auxlemmajaynes}). Let $y^*$ be an output of $\ORACLE(\rho)$. Let us apply Lemma \ref{finitaryJaynes} with $\pi := \sum_i y^*_i \ket{i}\bra{i} / \Vert y^* \Vert_1$ and ${\cal T} = \{ X := \sum_i \text{tr}(A_i \rho) \ket{i}\bra{i}, Y := \sum_i b_i \ket{i}\bra{i} \}$. Note that $\Delta({\cal T}) = R$. We find there is an integer $k \leq \gamma$ such that 
\begin{equation}
\tilde{\pi} := \frac{ \exp \left( - \frac{\kappa}{4 R^2} (k X + (\gamma-k) Y)   \right)    }{ \tr \left( \exp \left( - \frac{\kappa}{4 R^2} (k X + (\gamma-k) Y)  \right)  \right)    }
\end{equation}
satisfies 
\begin{equation} \label{closenessjaynes}
[\pi - \tilde{\pi}]_{\cal T} \leq \kappa. 
\end{equation}

Let $N$ be the integer which minimizes $|\kappa N' - \Vert y^* \Vert_1|$ over $N' \in \left[\left \lceil \frac{\alpha}{\kappa}  \right \rceil \right]$. By the bound $\Vert y^* \Vert \leq \sum_i b_i y_i^* \leq \alpha$, we have 
\begin{equation} \label{boundonnormdifferencey}
|\kappa N - \Vert y^* \Vert_1| \leq \kappa. 
\end{equation}

Then 

\begin{eqnarray}
\sum_{j=1}^m  \overline{q}_{\rho, k_{\text{opt}}}  \tr(A_{j} \rho) &\geq&   \sum_{j=1}^m \frac{y^*_j}{ \Vert y^* \Vert_1 }  \tr(A_{j} \rho) - (\kappa + \nu) \nonumber \\
&\geq&  \frac{1}{ \Vert y^* \Vert_1 } \tr(C \rho) - (\kappa + \nu) \nonumber \\
&\geq& \frac{1}{\kappa (N+1)}\tr(C \rho) - (\kappa + \nu),
\end{eqnarray}
where the first inequality follows from Eq.~(\ref{closenessjaynes}) and the fact that $\overline{q}_{\rho, k}$ is $\nu$-close to $q_{\rho, k}$, the second from Eq.~(\ref{kindoffeasible}), and the last from Eq.~(\ref{boundonnormdifferencey}).  

Likewise,
\begin{eqnarray}
 \sum_{j=1}^m   \overline{q}_{\rho, k_{\text{opt}}}  b_{i}  &\leq&  \sum_{j=1}^m  \frac{y^*_i}{\Vert y^* \Vert_1} b_{i} + R (\kappa + \nu)       \nonumber \\
&\leq&    \frac{1}{ \Vert y^* \Vert_1 } \alpha + R (\kappa + \nu)   \nonumber \\
&\leq&  \frac{1}{\kappa (N-1)}\alpha + R (\kappa + \nu).
\end{eqnarray}

\end{proof}

\section{Implementing the Oracle for $\overline{h}(\rho, \lambda, \mu)$}  \label{oracleforhbar}

In this section we explain how to implement the oracle that outputs the entries of the Hamiltonian $\overline{h}(\rho, \lambda, \mu)$ defined in Eq.~(\ref{hbar}). We start with the following standard result in quantum algorithms: 

\begin{lem}  \label{estimatingAi}
Given a $s$-sparse $n \times n$ Hermitian matrix $A$ with $\Vert A \Vert \leq 1$ and a density matrix $\rho$, with probability larger than $1 - p_{e}$, one can compute $\tr(\rho A)$ with additive error $\varepsilon$ in time $O(s \varepsilon^{-2} \log^4(ns/ (p_{e} \varepsilon)))$ using $O(\log(1/p_e) \varepsilon^{-2})$ copies of $\rho$. 
\end{lem}

\begin{proof}
Using the Hamiltonian simulation technique of Ref.~\cite{BCK15}, the phase estimation algorithm takes time $(s\log^4(ns/\varepsilon)))$ to measure to accuracy $\varepsilon/2$ the energy of $A_i$ in the state $\rho$. By the Chernoff bound, repeating the process  $O(1/\varepsilon^2)$ times allows us  to obtain an estimation for $\tr(\rho A_i)$ to accuracy $\varepsilon$.
\end{proof}

\begin{lem} \label{implementighbar}
Using $O(\log(m/p_e) h_{\text{precision}}^{-2})$ copies of $\rho$ and time  $O(s h_{\text{precision}}^{-2} \log^4(n m s/ (p_{e} h_{\text{precision}}))$, one can implement $O[\overline{h}]$ with error probability $p_e$. 
\end{lem}

\begin{proof}
Since by assumption we have an oracle for the $\{ b_i \}$, we can focus on showing how to compute $\text{tr}(A_i \rho)$ to accuracy $\nu$ given access to an oracle for the entries of the $A_i$ 
and copies of the state $\rho$. Lemma \ref{estimatingAi} shows how to compute, with probability at least $1 - p'_{e}$, an estimate of $\text{tr}(A_i \rho)$ to accuracy $h_{\text{precision}}$ in time $O(s h_{\text{precision}}^{-2} \log^4(ns/ (p'_{e} h_{\text{precision}}))$ using $O(\log(1/p'_e) h_{\text{precision}}^{-2})$ copies of $\rho$. Suppose the input of the oracle is $\sum_{i=1}^m c_i \ket{i}$. Then its output is $\sum_{i=1}^m c_i \ket{i, e_i}$ with $e_i$ the estimation of $\tr(A_i \rho)$. We know each $e_i$ is within $h_{\text{precision}}$ from the true value with probability at least $1 - p'_e$. Therefore by the union bound all of the $e_i$ are $h_{\text{precision}}$-close with probability $1 - m p'_e$. 
\end{proof}

\section{The Quantum Algorithm: Correctness and Complexity}   \label{proofmain}

We are ready to prove:

\begin{thm} [restatement of Theorem \ref{mainthm}]
Algorithm \ref{alg:q} runs in time 
\be
\tilde{O} \left( \frac{R^{21}}{\delta^{11}} G_{\overline{h}} G_M  \right) + \tilde{O} \left( \frac{R^{13}}{\delta^{5}} T_{\text{Meas}}  \right).
\ee
The algorithm fails with probability at most 
\be
O((R/\delta)^{18} (nm)^{10}\exp(- \log^{\xi}(nm))).
\ee
Assuming it does not fail, if it outputs $\mathsf{Larger}$, the optimal objective value is larger than $(1 - \delta)\alpha$. Otherwise it outputs a sample of a probability distribution $p$ and a real number $L$ such that $y = L p$ is dual feasible and $ \sum_i y_i b_i \leq (1 + \delta) \alpha$.
\end{thm}

\begin{proof}

\noindent \textit{Correctness of Algorithm:}  We let $p_e = \exp(- \log^{\xi}(nm))$, with $p_e$ the error probability for the oracle $\overline{h}$. If at some call to $\overline{h}$, the output is wrong, we declare the algorithm failed. Let us first assume that the oracle $\overline{h}$ outputs the correct value in all calls. Later we will show the oracle to $\overline{h}$ is used at most $O((R/\delta)^{10} (nm)^{10}\exp(- \log^{\xi}(nm)))$ times, the algorithm fails due to a faulty $\overline{h}$ with probability at most $O((R/\delta)^{18} (nm)^{10}\exp(- \log^{\xi}(nm)))$. 

Let us first consider the case in which for every $t \leq T$, after step 2, $y^{(t)}$ is not equal to $(0, \ldots, 0)$. Then the algorithm has to output a sample from $y / \Vert y \Vert_1$ and the value of $\Vert y \Vert_1$, with 
\be
\overline{y} = \frac{\delta \alpha}{2R} e_1 + \frac{1}{T} \sum_{t=1}^T y^{(t)}.
\ee
Note $\Vert y^{(t)} \Vert_1 = \varepsilon N_t$, so we know all of them. Therefore we can compute $\Vert \overline{y} \Vert_1$. Since we have samples from $y^{(t)}$ for all $t \leq T$, we can sample from $\overline{y}$, which is a convex combination of them (and where we know the mixing probability distribution). 

Lemma \ref{implementighbar} with $h_{\text{precision}} = \varepsilon/2$ and Lemma \ref{algorthmfororacleworks} with $\nu = \kappa = \varepsilon/2$ show Step 2 of the algorithm implements $\ORACLE(\rho^{(t)})$ as in Algorithm \ref{alg:implementingoracle}. Then for each $t$, Step 2 outputs a vector $y^{(t)}$ such that, with probability at least $1 - \exp(-\log^{\xi}(n m))$,
\begin{equation}  \label{approxobjective}
\sum_{i=1}^m b_i y^{(t)}_i \leq \alpha(1 + 2R \varepsilon),
\end{equation}
and
\begin{equation}    \label{approxconstriantdualonrho}
\sum_{i=1}^m y^{(t)}_i \tr(A_i \rho^{(t)}) \geq \tr(C \rho^{(t)}) - 2  \alpha \varepsilon.
\end{equation}

The remaining steps of the algorithm run the Matrix Multiplicative Weight method. One difference (which will be important to keep the quantum complexity of the algorithm low), is the sparsification of the pay-off matrix $M^{(t)}$. Let 
\begin{equation}
\hat{M}^{(t)} := \left( \sum_{i=1}^m y^{(t)}_i A_{i} - C + 2 \alpha I \right)/4 \alpha.
\end{equation}

Since
\be    \label{boundondev2}
\left \Vert  \sum_{i=1}^m y^{(t)}_i A_{i} - C  \right \Vert \leq \sum_{i=1}^m y^{(t)}_i + 1 \leq  \sum_{i=1}^m b_i y^{(t)}_i + 1 \leq \alpha + 1,
\ee
we have $0 \leq \hat{M}^{(t)}   \leq I$.

By the Matrix Hoeffding bound (Lemma \ref{hoeffdingbound}), with probability larger than $1 - \exp(-\log^{\xi}(n m))$, 
\begin{equation}
\left \Vert  \hat{M}^{(t)}  - M^{(t)}  \right \Vert \leq \frac{1}{4T}.
\end{equation}
Then by Lemma \ref{poulinwocjan} and the fact the Gibbs sampler has error $\varepsilon/2$, we find
\be  \label{rhosareclose}
\left \Vert \rho^{(t)} -   \hat{\rho}^{(t)}    \right \Vert_1  \leq  \varepsilon,
\ee
with
\begin{equation}
 \hat{\rho}^{(t)} :=  \frac{\exp(- \varepsilon'((\sum_{\tau=1}^t \hat{M}^{(\tau)}))}{\tr \left( \exp(- \varepsilon'((\sum_{\tau=1}^t \hat{M}^{(\tau)})) \right)}.
\end{equation}

We are ready to show the correctness of the algorithm. Since $\eta = \frac{\delta \alpha}{2R}$, we find
\begin{eqnarray}
\overline{y}.b &=& R \eta + \frac{1}{T} \sum_{t=1}^T y^{(t)}.b \nonumber \\
&\leq& R \eta + \alpha(1 + 2R \varepsilon) \nonumber \\
&=& \alpha \left(1 + \delta/2 +  2R \varepsilon  \right)   \nonumber \\
&\leq& (1 + \delta)\alpha. 
\end{eqnarray}

Let us now show that $\overline{y}$ is dual feasible. It is clear that $\overline{y} \geq 0$. In analogy with Eq.~(\ref{mmwpsd}), we have
\begin{eqnarray}
 && \lambda_{n}\left( \sum_{j=1}^m \overline{y}_j A_j - C \right) \\ &=&  \lambda_{n}\left(  \frac{1}{T} \sum_{t=1}^T \sum_{j=1}^m y_j^t  A_j - C \right) + \eta  \\
 &=& 4 \alpha \lambda_{n}\left(  \frac{1}{T} \sum_{t=1}^T \left(  \sum_{j=1}^m y_j^t  A_j - C + 2 \alpha I \right)/4 \alpha    \right) - 2 \alpha + \eta   \\
 &\overset{(i)}{\geq} &  \frac{4 \alpha}{(1 + \varepsilon)} \frac{1}{T} \sum_{t=1}^T  \text{tr}\left( \hat{\rho}^{(t)} \left(  \sum_{j=1}^m y_j^t  A_j - C + 2 \alpha I \right)/4 \alpha    \right) - \frac{4 \alpha \ln(n)}{T(1 + \varepsilon)\varepsilon} - 2 \alpha + \eta  \label{beforebeforelast} \\
 &\overset{(ii)}{\geq} &  \frac{4 \alpha}{(1 + \varepsilon)} \frac{1}{T} \sum_{t=1}^T  \text{tr}\left( \rho^{(t)} \left(  \sum_{j=1}^m y_j^t  A_j - C + 2 \alpha I \right)/4 \alpha    \right)  - \frac{4 \alpha \varepsilon }{ (1 + \varepsilon)}    - \frac{4 \alpha \ln(n)}{T(1 + \varepsilon)\varepsilon} - 2 \alpha + \eta   \label{beforebeforelastnew} \\
&\overset{(iii)}{\geq}& \frac{4 \alpha}{(1 + \varepsilon)} \left( \frac{1}{2} - 2\alpha \varepsilon  \right)   - \frac{4 \alpha \varepsilon}{(1 + \varepsilon)}  - \frac{4 \alpha \ln(n)}{T(1 + \varepsilon)\varepsilon} - 2 \alpha + \eta  \label{beforelast} \\
&=& - \frac{6 \alpha \varepsilon}{(1 + \varepsilon)} - \frac{8 \alpha^2 \varepsilon}{(1 + \varepsilon)} + \frac{\alpha \delta}{2R} -  \frac{4 \alpha \ln(n)}{T(1 + \varepsilon)\varepsilon}  \nonumber \\
&\geq& \frac{\delta \alpha}{4 R} - \frac{4 \alpha \ln(n)}{T(1 + \varepsilon)\varepsilon}  \\ 
&\geq& 0.
\end{eqnarray}
where we used that $\alpha \geq 1$,
\be
T = \frac{16 R \ln(n)}{\delta \varepsilon},
\ee
and
\be
\varepsilon = \frac{\delta}{28 R^2}. 
\ee

Inequality (i) follows from the Matrix Multiplicative Weight method (Lemma \ref{MMWM}), which can be applied since by Eq.~(\ref{boundondev2}), 
\be
0 \leq \left(  \sum_{j=1}^m y_j^t  A_j - C + 2 \alpha I \right)/4 \alpha  \leq I.
\ee

Inequality (ii) follows from Eq.~(\ref{rhosareclose}), while Inequality (iii) follows from Eq.~(\ref{approxconstriantdualonrho}). 

Let us now turn to the case in which there is a $t \in [T]$ such that the loop in Step 2 outputs $y^{(t)} = (0, \ldots, 0)$. Then by the Chernoff bound and Lemma \ref{finitaryJaynes} we find that for every $y \geq 0$ s.t. 
\be 
\sum_{i=1}^m y_i \tr(A_i \rho) \geq \tr(C \rho) - 3 \alpha \varepsilon,
\ee
we must have
\be
y . b \leq \alpha(1 - 3 R \varepsilon).
\ee
By duality of linear programming it follows the maximum over $\lambda \geq 0$ of 
\be
\lambda (\tr(C \rho^{(t)}) - 3\alpha \varepsilon) 
\ee
subject to the constraints 
\be
\lambda \tr(A_i \rho^{(t)}) \leq b_i, \hspace{0.2 cm} i \in [m]
\ee
must be larger than $\alpha(1 - 3 R \varepsilon)$. Note that $\lambda = \lambda \tr(\rho^{(t)}) \leq R$. Then defining $X = \lambda_{\text{optimal}} \rho^{(t)}$ (with $\lambda_{\text{optimal}}$ the optimal value of $\lambda$ for the LP above), we find that $\tr(A_i X) \leq b_i$ for all $i \in [m]$ and 
\be
\tr(C X) \geq \alpha(1 - 3 R \varepsilon)  \geq (1 - \delta)\alpha. 
\ee

Finally, let us bound the probability that the algorithm fails. This can be due to three causes. The first is a faulty oracle call for $\overline{h}$. This is upper bounded by 
\be \label{errors}
O((R/\delta)^{18} (nm)^{10}\exp(- \log^{\xi}(nm))). 
\ee
The second is due to an error in estimating the values of $\{ \tr(A_i \rho^{(t)}) \}$ in Step 2 of the algorithm. The third is the random sampling used to construct $M^{(t)}$. By the union bound the error probability of both are also bounded by Eq. (\ref{errors}). 

\vspace{0.2 cm}

\noindent \textit{Run-time Analysis:} We remind the reader we are assuming that $\alpha \geq 1$.

The cost of running step 2 is $\tilde{O}(\alpha R^2 / \varepsilon^3) \leq \tilde{O}(R^{9}/ \delta^3)$ times the sum of cost of the following: (i) performing $M$ times the Gibbs sampling of the Hamiltonian $\overline{h}$ with cost $M G_{\overline{h}}$, (ii) sampling $M$ times from the resulting distribution with cost $O(M)$, and (iii) computing estimates to the expectation values of the input matrices on the state $\rho^{(t)}$, with cost $(M+1)T_{\text{Meas}}$. Therefore the total cost of step 2 (for a particular $t \in [T]$) is
\be
\tilde{O}(  \frac{R^9}{\delta^3} M  \left( G_{\overline{h}} +  T_{\text{Meas}} \right)) = \tilde{O}((   \frac{R^{13}}{\delta^5}  \left( G_{\overline{h}} +  T_{\text{Meas}} \right))
\ee

The cost of step 4 is $(Q+1) G_{\overline{h}} = \tilde{O}( (R^{6}/\delta^4) G_{\overline{h}}$. The cost of step 5 is $Q = \tilde{O}(R^{6}/\delta^4)$. The cost of step 6-7 is 
\be
C_t G_M = \tilde{O}\left( \frac{\alpha}{\varepsilon^3} \frac{R^2}{\varepsilon^2} \left( \frac{1}{\varepsilon^2} + \frac{R^{6}}{\delta^4} \right) G_{\overline{h}}\right) G_M + \tilde{O} \left( \frac{2 \alpha}{\varepsilon^3} \frac{R^2}{\varepsilon^4} \right)G_M =  \left(\frac{R^{19}}{\delta^{9}} \right) G_{\overline{h}} G_M.
\ee

As each step is repeated $T = \tilde{O}(R^3/\delta^2)$ times, the total cost is
\be
\tilde{O} \left( \frac{R^{21}}{\delta^{11}} G_{\overline{h}} G_M  \right) + \tilde{O} \left( \frac{R^{13}}{\delta^{5}} T_{\text{Meas}}  \right).
\ee

\end{proof}

\begin{lem}  \label{hoeffdingbound}
(Matrix Hoeffding Bound, Theorem 2.8 of \cite{Tropp10}) Suppose $Z_1, \ldots, Z_k$ are independent random $d \times d$ Hermitian matrices satisfying $\mathbb{E}[Z_i] = 0$ and $\Vert Z_i \Vert \leq \lambda$. Then 
\begin{equation}
\text{Pr} \left [   \left \Vert  \frac{1}{k} \sum_{i=1}^k Z_i  \right \Vert \geq \delta  \right ] \leq d . e^{- \frac{k \delta^2}{8 \lambda^2}}.
\end{equation}
\end{lem}

\begin{lem}   \label{poulinwocjan}
Let $H, H'$ be Hermitian matrices. Then 
\begin{equation}
\left \Vert \frac{e^{ H}}{\tr(e^{H})}  - \frac{e^{H'}}{\tr(e^{H'})}  \right \Vert_1 \leq   2\left( e^{\Vert H - H' \Vert} - 1   \right).
\end{equation}
\end{lem}
\begin{proof}
We can assume w.l.o.g. that $\tr(e^H) \geq \tr(e^{H'})$. 

Let $M$ be an arbitrary operator with $\Vert M \Vert \leq 1$. We write
\begin{equation}
\tr \left( M \frac{e^H}{\tr(e^H)}  \right) =  \frac{\tr(e^{H'})}{\tr(e^H)}  \tr \left ( M   \frac{e^{H'}}{\tr(e^{H'})} \left(  {\cal T} \exp \left( \varepsilon  \int_{0}^1 dt e^{-tH'} (H-H') e^{tH'} \right)  \right)   \right)
\end{equation}
with ${\cal T}$ the time-ordered operator, i.e. 
\begin{eqnarray}
&&{\cal T} \exp \left(  \int_{0}^1 dt e^{-tH'} (H-H') e^{tH'} \right) \\ &:=& I + \int_{0}^1 dt e^{-t H'} (H-H') e^{t H'} + \int_{0}^1 dt_1 \int_{0}^{t_1} dt_2 e^{-t_1 H'} (H-H') e^{(t_1 - t_2)H'}(H-H')e^{t_2 H'}  + \ldots  \nonumber,
\end{eqnarray}
where the times are such that $1 \geq t_1 \geq \ldots \geq t_k$ .

We now follow closely the argument in the Appendix of \cite{Has16}. Write
\begin{equation}
 \frac{\tr(e^{H})}{\tr(e^{H'})} \tr \left( M \frac{e^H}{\tr(e^H)}  \right) = \sum_{k=0}^{\infty} T_k,
\end{equation}
with
\begin{eqnarray}
&&T_k \\ &:=&  \tr \left( M \frac{e^{H'}}{\tr(e^{H'})} \left(  \int_{0}^1 dt_1 \int_{0}^{t_1} dt_2 \ldots  \int_{0}^{t_k} dt_k   e^{-t_1 H'} (H-H') e^{(t_1 - t_2)H'}(H-H') \ldots (H-H') e^{t_k H'}  \right)      \right)  \nonumber
\end{eqnarray}
Note 
\begin{equation}
T_0 = \tr \left( M \frac{e^{H'}}{\tr(e^{H'})} \right).
\end{equation}

Consider the $k$-th term in the series (for $k \geq 1$). We can bound it by
\begin{eqnarray}
T_k &\leq&  \frac{1}{\tr(e^{H'})}  \int_{0}^1 dt_1 \int_{0}^{t_1} dt_2 \ldots  \int_{0}^{t_k} dt_k    \\ &&    \left \Vert  e^{(1-t_1 )H'}(H-H')e^{(t_{1} - t_{2})H'}(H-H') \ldots (H-H')e^{(t_{k-1} - t_k)H'}(H-H')e^{t_k H'}   \right \Vert_1   \nonumber
\end{eqnarray}
We note there are $k$ terms equal to $H-H'$ and $k$ terms given by exponentials $\exp(\delta_i H')$, for positive $\delta_i$. H\"older's inequality give $\Vert  X_1 \ldots X_l   \Vert_1 \leq  \prod_i \Vert X_i \Vert_{p_i}$, for the $p_i$ norms of the matrices, with $\sum_i p_i^{-1} = 1$. We apply it to the expression above, taking the $p = \infty$ norm for the $(H-H')$ terms and the $1/\delta_i$ norm for the $e^{\delta_i H'}$ terms. Since $\sum_i \delta_i = 1$, we can apply the inequality. We find
\begin{equation}
T_k \leq \Vert H - H' \Vert^k   \int_{0}^1 dt_1 \int_{0}^{t_1} dt_2 \ldots  \int_{0}^{t_k} dt_k  \leq \frac{\Vert H - H' \Vert^k }{k!}.
\end{equation}
Therefore
\begin{equation}
\left | \frac{\tr(e^{H})}{\tr(e^{H'})}  \tr \left( M \frac{e^H}{\tr(e^H)}  \right) -    \tr \left( M \frac{e^{H'}}{\tr(e^{H'})} \right)  \right | \leq \sum_{k=1}^{\infty} |T_k| \leq e^{\Vert H - H' \Vert} - 1
\end{equation}

The result follows from the Golden-Thompson inequality, which implies
\be
\tr(e^H) \leq \tr(e^{H'})e^{\Vert H - H' \Vert}. 
\ee


\end{proof}

Finally let us prove
\begin{cor} [Restatement Corollary \ref{cormain}]%
Using the Gibbs Sampler from Ref. \cite{PW09}, Algorithm \ref{alg:q} runs in time $\tilde{O}( n^{\frac{1}{2}} m^{\frac{1}{2}} s^2 R^{32}  / \delta^{18})$.
\end{cor}

\begin{proof}

The result is a consequence of Theorem \ref{mainthm} and the main result of \cite{PW09}, which gives a method to prepare an $\varepsilon$-approximation to $e^{\beta H}/\tr(e^{\beta H})$ for a $s'$-sparse $H$ with $\Vert H \Vert \leq 1$ using
\be
\tilde{O}(\sqrt{\text{dim}(H)} \beta s'/\varepsilon).
\ee
calls to the oracle and two-qubit gates.\footnote{The result of Ref.  \cite{PW09} is presented in the special case of a local Hamiltonian. However one can check that the only property of the Hamiltonian needed is that $U(t) = e^{-i t H}$ can be implemented to error $\varepsilon$ by a circuit of size $t \poly(n, 1/\varepsilon)$. Since by \cite{BCK15} $s$-sparse Hamiltonians can be implemented by a circuit of size $st\poly(n, \log(1/\varepsilon))$, the result can be applied to them.}

We apply this Gibbs sampler to steps 2, 4 and 7 of the quantum algorithm. In order to estimate the running time of each of these applications, we must estimate the associated $\beta$ and sparsity. In steps 2 and 4, the associated $\beta$ is upper bounded by $O(1/\varepsilon) = O(R^2/\delta)$ and the sparsity is one. In step 7 the associated $\beta$ is upper bounded by $\varepsilon' T = \tilde{O}(R/\delta)$, while the sparsity is upper bounded by $TQ$ times the sparsity of each of the input matrices, which gives $\tilde{O}(s R^{9}/ \delta^6)$. Finally, since to query an element the Hamiltonian we need to query each of the $A_i's$ matrices appearing int he decomposition, we have a cost of $Ts$ for querying an element of the Hamiltonian. Therefore we have the bounds
\be
G_{\overline{h}} \leq \tilde{O}( \sqrt{m}R^2/\delta)
\ee
and
\be
G_{M} \leq \tilde{O}( \sqrt{n} s^2 R^{9}/\delta^{6}).  
\ee

\end{proof}

\section*{Acknowledgments}
We thank Joran van Apeldoorn, Ronald de Wolf, Andras Gilyen, Aram Harrow, Sander Gribling, Matt Hastings, Cedric Yen-Yu Lin, Ojas Parekh, and David Poulin for interesting discussions and useful comments on the paper. This work was funded by Cambridge Quantum Computing, Microsoft and the National Science Foundation.

\appendix

\section{Reduction to $b_i \geq 1$}   \label{appendix}

Here we prove:

\begin{lem}[Restatement Lemma \ref{reductionpositiveb}]
One can sample from a $\delta$-optimal solution of the SDP given by Eq. (\ref{dualproblem}) (with dimension $n$, $m$ variables, size parameter $R$ and upper bound on optimal solution vector $r$) given the ability to sample from a $\delta/r$-optimal solution of the SDP given by Eq.~(\ref{dualproblem}) (with dimension $n+1$, $m+1$ variables and size parameter $2R+1$) in which $b_i \geq 1$ for all $i \in [m]$.
\end{lem}

\begin{proof}
Given the SDP of Eq. (\ref{dualproblem}), we define the following related SDP with $m+1$ variables in $n+1$ dimensions:
\begin{eqnarray}  \label{extendeddualSDP}
&& \min \sum_{i=1}^{m} b_i y_i + (R+1)\sum_{i=1}^{m+1} y_i    \nonumber \\
&&  \sum_{i=1}^m  y_i  \left[ {\begin{array}{cc}  A_i & 0 \\       0 & 1       \end{array} } \right] + y_{m+1}    \left[ {\begin{array}{cc}  0 & 0 \\       0 & 1       \end{array} } \right]      \geq \left[ {\begin{array}{cc}  C & 0 \\       0 & r       \end{array} } \right]    \nonumber \\
&& y \geq 0,
\end{eqnarray}
for $r$ an upper bound on $\sum_i z_i$, for an optimal solution $(z_1, \ldots, z_m)$ of the SDP given by Eq. (\ref{dualproblem}).

Note that since $\max |b_i| = R$, the vector defining the objective function of the SDP of Eq. (\ref{extendeddualSDP}) has all elements larger or equal than one. 

Let $(y_1,\ldots, y_{m+1})$ be a $\delta$-optimal solution of the SDP given by Eq. (\ref{extendeddualSDP}). We claim $(y_1, \ldots, y_m)$ is an $\delta$-optimal solution of the original SDP given by Eq. (\ref{dualproblem}). Indeed we have that
\be
\sum_{i=1}^m y_i A_i \geq C
\ee
so $(y_1, \ldots, y_m)$ is feasible. It remains to show that 
\be
\sum_{i=1}^m b_i y_i \leq \text{opt} + \delta,
\ee
with $\text{opt}$ the optimal value of the SDP of Eq. (\ref{dualproblem}). Suppose it was not the case and that
\be
\sum_{i=1}^m b_i y_i > \text{opt} + \delta.
\ee
Let us find a contradiction.

Let $(z_1, \ldots, z_m)$ be the optimal solution to the SDP of Eq. (\ref{dualproblem}) defined above and consider the following solution to the SDP given by Eq. (\ref{extendeddualSDP}):
\be
(y'_1, \ldots, y'_{m+1}) := \left(z_1, \ldots, z_m, \sum_{i=1}^{m+1} y_i - \sum_{i=1}^m z_i \right).
\ee
Note that since 
\be
\sum_{i=1}^{m+1} y_i \geq r \geq \sum_{i=1}^{m} z_i , 
\ee
it follows
\be
\sum_{i=1}^{m+1} y_i - \sum_{i=1}^m z_i  \geq 0,
\ee 
so all the elements are non-negative. Moreover, the matrix inequality constraint is satisfied since 
\be
\sum_{i=1}^m y'_i A_i = \sum_{i=1}^m z_i A_i \geq C 
\ee
and
\be
\sum_{i=1}^{m+1} y'_i = \sum_{i=1}^{m+1} y_i \geq r.  
\ee
Finally since
\be
\sum_{i=1}^{m} b_i y'_i  =  \sum_{i=1}^{m} b_i z_i,
\ee
we find
\be
\left( \sum_{i=1}^{m} b_i y'_i + (R+1)\sum_{i=1}^{m+1} y'_i \right)  - \left( \sum_{i=1}^{m} b_i y_i + (R+1)\sum_{i=1}^{m+1} y_i\right) = \text{opt} -  \sum_{i=1}^{m} b_i y_i < - \delta,
\ee
which contradicts the assumption that $(y_1, \ldots, y_{m+1})$ is $\delta$-optimal.

Finally note that $C$ has norm $\Vert C \Vert = \max(1, r)$ which might be larger than one. Therefore we solve an associated SDP with a rescaled $C$ by $1/r$. To solve the original SDP with additive error $\delta$, we must solve the rescaled SDP with additive error $\delta/r$.

\end{proof}

\end{document}